\documentclass[conference]{IEEEtran}
\usepackage[left=1.62cm,right=1.62cm,top=1.85cm,bottom=2.6cm]{geometry}

\usepackage{amsthm,bm,amssymb,graphicx,multirow,amsmath,color,amsfonts}
\usepackage[update,prepend]{epstopdf}%
\usepackage[noadjust]{cite}
\usepackage{cleveref}
\usepackage{amsmath}
\usepackage{tikz}
\usetikzlibrary{arrows,calc}		
\usepackage{bbm} 
\usepackage{p dfpages}
\usepackage{tabulary}
\usepackage{multirow}
\usepackage{comment}
\usepackage{dsfont}
\usepackage{pgf,tikz}
\usepackage{paralist}
\usepackage{amsfonts}
\usepackage{setspace,lipsum}
\usepackage{url}
\usepackage{graphicx}
\usepackage{comment}
\usepackage{bigints}
\usepackage{mathtools}
\usepackage[english]{babel}
\usepackage[autostyle]{csquotes}
\usepackage{textcomp}
\usepackage{varwidth}
\usepackage{tcolorbox}
\usepackage{lipsum}
\usepackage{cite} 

\usepackage{amssymb,amsmath,amsfonts,amsthm}
\usepackage{mathrsfs}
\usepackage{cite}
\usepackage{array}
\usepackage{tabularx}
\usepackage{color}
\usepackage{mathtools}
\usepackage{subfigure}
\usepackage{mathtools}
\usepackage{tikz,pgf}
\usetikzlibrary{calc,positioning,mindmap,trees,decorations.pathreplacing}
\usepackage{graphicx}
\usepackage{bbm}
\usepackage[utf8]{inputenc}
\usepackage{algorithm}
\let\OLDthebibliography\thebibliography
\renewcommand\thebibliography[1]{
  \OLDthebibliography{#1}
  \setlength{\parskip}{0pt}
  \setlength{\itemsep}{0pt plus 0.3ex}
}

\newcommand\numberthis{\addtocounter{equation}{1}\tag{\theequation}}

\usepackage{amssymb,amsmath,amsfonts,amsthm}
\usepackage{mathrsfs}
\usepackage{cite}
\usepackage{array}
\usepackage{tabularx}
\usepackage{color}
\usepackage{mathtools}
\usepackage{subfigure}
\usepackage{mathtools}
\usepackage{tikz,pgf}
\usetikzlibrary{calc,positioning,mindmap,trees,decorations.pathreplacing}
\usepackage{graphicx}
\usepackage{bbm}
\usepackage[utf8]{inputenc}
\usepackage{algorithm}

\renewcommand{\IEEEQED}{\IEEEQEDopen}

\newcommand\blfootnote[1]{%
  \begingroup
  \renewcommand\thefootnote{}\footnote{#1}%
  \addtocounter{footnote}{-1}%
  \endgroup
}

\def\nbe{{\mathbf{e}}}
\def\nbf{{\mathbf{f}}}
\def\nbg{{\mathbf{g}}}
\def\nbh{{\mathbf{h}}}

\def\nbp{{\mathbf{p}}}
\def\nbq{{\mathbf{q}}}

\def\nbv{{\mathbf{v}}}
\def\nbw{{\mathbf{w}}}

\def\nb0{{\mathbf{0}}}
\def\nb1{{\mathbf{1}}}

\def\nbE{{\mathbf{E}}}

\def\nbG{{\mathbf{G}}}
\def\nbH{{\mathbf{H}}}

\def\nbZ{{\mathbf{Z}}}

\def\ncalC{{\mathcal{C}}}

\def\ncalN{{\mathcal{N}}}

\def\nbbC{{\mathbb{C}}}

\def\nbbE{{\mathbb{E}}}

\def\nbbP{{\mathbb{P}}}

\newtheorem{thm}{Theorem}

\newtheorem{cor}{Corollary}

\allowdisplaybreaks 

\usepackage{setspace}	

\newcommand{\subparagraph}{} 
\usepackage{titlesec}        

\begin{document}

\title{Optimal Beamforming and Outage Analysis for Max Mean {\rm SNR} under RIS-aided Communication}
	\author{Kali Krishna Kota, Praful D. Mankar, Harpreet S. Dhillon
	}

	\maketitle
	\thispagestyle{empty}
\pagestyle{empty}
	
\begin{abstract}

This paper considers beamforming for a reconfigurable intelligent surface (RIS)-aided multiple input single output (MISO) communication system in the presence of Rician multipath fading. Our  aim is to jointly optimize the transmit beamformer and RIS phase shift matrix for maximizing the mean signal-to-noise (SNR) of the combined signal received over direct and indirect links. While numerical solutions are known for such optimization problems, this is the first paper to derive closed-form expressions for the optimal beamformer and the phase shifter for a closely related problem. In particular, we maximize a carefully constructed lower bound of the mean SNR, which is more conducive to analytical treatment. Further, we show that effective channel gain under optimal beamforming follows Rice distribution. Next, we use these results to characterize  a closed-form expression for the outage probability under the proposed beamforming scheme, which is subsequently  employed to derive an analytical expression for the ergodic capacity. Finally, we numerically demonstrate the efficacy of the proposed  beamformer solution in comparison with the existing algorithmically  obtained optimal solution for the exact mean SNR maximization.    
\end{abstract}
\IEEEpeerreviewmaketitle 
\blfootnote{K.K. Kota and P. D. Mankar  are with SPCRC, IIIT Hyderabad, India (Email: kali.kota@research.iiit.ac.in,   praful.mankar@iiit.ac.in). H. S. Dhillon is with Wireless@VT, Department of ECE, Virginia Tech, Blacksburg, VA (Email:  hdhillon@vt.edu). H. S. Dhillon gratefully acknowledges the support of US NSF (grant ECCS-2030215). }
\section{Introduction}
RIS is a planar array that consists of many sub-wavelength-sized elements formed of meta-materials. One can control the phases of these elements to change the way they interact with the impinging electromagnetic wave, thereby controlling the local propagation environment to a certain extent \cite{Emil}. As a result, this ability can influence key characteristics of the propagation environment, such as reflection, refraction, and scattering, which were thus far assumed to be uncontrollable in wireless communications systems~\cite{1}. If configured properly, RISs can reduce the effect of fading by redirecting the impinging signals such that they add constructively at the receiver. The advantages of such a technology are low power consumption, high spectral efficiency, improved coverage, and better reliability \cite{2}. Even though the idea of using RIS is conceptually similar to several well-studied technologies, such as the relay transmission, backscatter communication, and massive MIMO, a key differentiating feature is its low-cost implementation and lower power consumption, especially when all the elements are passive (which will be our assumption in this paper). Additionally, RISs operate in the full-duplex model by design without additional hardware requirements, which is not the case in the competing technologies mentioned above~\cite{4}.

Needless to say, including RIS in the communication system design comes with its own unique challenges. Two particularly important ones are the  channel estimation and the joint design of transmit-receive beamformers and the RIS phase shift matrix. The challenges in channel estimation stem from  \cite{5}: 1) the high dimensionality of the channel matrix, and 2) the lack of active elements in passive RISs for aiding the channel estimation process \cite{6}. Therefore, it is not always reasonable to assume perfect CSI at the transmitter (particularly for the channels involving RISs) for optimal beamforming. Moreover, instantaneous CSI feedback and instantaneous RIS reconfiguration will increase the system overhead significantly; as a result, the ability to reconfigure it quickly for precise RIS design will be complex.
 Due to these reasons, it is more practical to assume the statistical knowledge of CSI to optimize RIS-aided communications systems, which is the main theme of this paper.

{\em Related Works:} The authors of \cite{7,8,9} investigate the joint optimal instantaneous transmit beamformer and phase shift matrix design for RIS-aided communication systems while assuming  perfect CSI at the transmitter. Authors in \cite{7,8} propose iterative algorithms such as fixed point iteration and manifold optimization, and branch-and-bound techniques to maximize the {\rm SNR}. Further, \cite{9} proposes a new sum-path-gain maximization criterion to obtain a  suboptimal solution, which is numerically shown to achieve near-optimal RIS-MIMO channel capacity. 
While these works provided useful design insights, they all considered instantaneous CSI knowledge at the transmitter, which is not always feasible. 
Inspired by this, \cite{10,11,12,hu2020statistical} aim to jointly optimize the transmit beamformers and phase shifter using statistical CSI. The authors in \cite{10} maximize the sum rate for a multi-user MIMO system in the presence of spatial correlation. 
The authors of \cite{11} and \cite{12} derive closed-formed approximate expressions for the sum rate and maximize it with respect to the transmit beamformer and RIS phase shifts in downlink and uplink, respectively. 
A fractional programming-based iterative solution is presented in \cite{11}, whereas \cite{12} proposes a genetic algorithm-based solution. Finally, \cite{hu2020statistical} presents an optimal solution that maximizes the mean SNR, or equivalently the upper bound on ergodic capacity, for the RIS-aided MISO system. The solutions are obtained by alternating over the transmit beamformer and phase shift matrix sub-problems for which closed-form expressions are obtained. 
However, such numerical solutions often prohibit further analytical investigation, such as the understanding of the exact functional dependence of the optimal solutions on key system parameters. {\em Therefore, it is equally, or perhaps even more, important to obtain optimal solutions in closed form, which will be our objective in this paper.} 

It is worth noting that the statistical CSI-based design of RIS-aided systems not only circumvents the system design complexity issues but often also makes the analysis tractable while resulting in comparable performance to the instantaneous CSI-based beamforming. 
The performance of the beamformer scheme is usually characterized using the outage analysis. Along these lines, the authors of \cite{14,15,16} analyze the outage performance for RIS-aided communication systems. 
The authors of \cite{14} derive a closed-form expression for the asymptotic outage probability, whereas \cite{15} obtain an approximate channel gain distributions for two cases: 1) RIS-aided system and 2) RIS-at-transmitter system. In addition, \cite{16} also derives a closed-form expression for a tight approximation of the channel gain and provides insight into diversity gain for RIS-aided communication systems.

{\em Contributions:} 
In this paper, we consider a RIS-aided MISO system where the transmitter performs optimal beamforming based on the statistical {\rm CSI} knowledge. While the mean {\rm SNR}-based maximization problem has been numerically solved for this setting in \cite{hu2020statistical}, our main technical contribution is in deriving closed-form expressions for the jointly optimal transmit beamformer and the RIS phase shift matrix for  maximizing a carefully-constructed lower-bound on the mean {\rm SNR}. For this optimal beamformer, we also show that the effective channel gain is Rice distributed, which results in a closed-form expression for the outage probability as well as an analytical expression for the ergodic capacity. Using numerical comparisons, we further demonstrate that the proposed beamforming method provides a tight lower bound on the ergodic capability and a tight upper bound on the outage probability compared to the numerically-obtained solution of the exact mean {\rm SNR} in \cite{hu2020statistical}. 

{\em Notations:} The notations $a^*$ and $|a|$ represent the conjugate and absolute value of $a$. $\left\lVert \mathbf{a} \right\lVert$ is the norm of vector $\mathbf{a}$, whereas $\mathbf{A}^T$, $\mathbf{A}^H$,  $\mathbf{A}_{i,:}$, $\mathbf{A}_{:,i}$ and $\mathbf{A}_{ij}$ are the transpose, Hermitian, $i$-th row, $i$-th column and $ij$-th element of the matrix $\mathbf{A}$, respectively. The notation $\mathbb{C}^{M\times N}$ is the set of  $M \times N$ complex matrices and ${\rm I_M}$ is $M\times M$ identity matrix. $\mathbb{E}[\cdot]$ is the statistical expectation operator and $\mathcal{CN}(\mathbf{\mu},\mathbf{K})$ denotes complex circular Gaussian distribution with mean $\mathbf{\mu}$ and covariance matrix $\mathbf{K}$.
\section{System Model}
This paper considers a RIS-assisted MISO communication system consisting of a transmitter equipped with $M$ antennas, a RIS with $N$ elements, and a single antenna receiver.  We assume that the receiver receives its signal through two links: 1) the direct link from the transmitter and 2) the indirect link through RIS. We assume that the line-of-sight (LoS) components exist in both of these links. To model  multipath fading for such  a scenario, we consider  Rician fading with factor $K$.  Thus, the distribution of the channel state (or CSI) is determined by the LoS components. We assume that the   knowledge of these LoS components is available at the transmitter for  determining the optimal choice of the transmit beam vector and  RIS phase shift matrix. 

The direct link channel coefficient is modeled as
\begin{equation}
    \nbg=\kappa_{\rm l}\bar{\nbg} + \kappa_{\rm n}\tilde{\nbg}, \label{eq:g_model}
\end{equation}
where  $\kappa_{\rm l}=\sqrt{\frac{K}{1+K}}$, $\kappa_{\rm n}=\sqrt{\frac{1}{1+K}}$,  $\tilde{\nbg}\sim\ncalC\ncalN(0,{\rm I_M})$,     $\bar{\nbg}=[1,~e^{j\theta_{\rm DD}},~~ \dots~~,~e^{j(M-1)\theta_{\rm DD}}]^T$, and $\theta_{\rm DD}$ is the departure angle of the direct link.
Similarly, the channel coefficients for transmitter-RIS link and RIS-receiver links are modeled as 
\begin{align*}
     \nbH=\kappa_{\rm l}\bar{\nbH}+\kappa_{\rm n}\tilde{\nbH} \text{~~and~~}
     \nbh=\kappa_{\rm l}\bar{\nbh}+\kappa_{\rm n}\tilde{\nbh},\numberthis \label{eq:hH_model}
\end{align*}
respectively, where $\bar{\nbh}=[1,~e^{j\theta_{\rm DI_2}},~~\dots~~,~e^{j(N-1)\theta_{\rm DI_2}}]^H$ and  $\bar{\nbH}_{nm}=e^{j(n\theta_{\rm AI_1}-m\theta_{\rm DI_1})}$ for $n=0,\dots,N-1$ and $m=0,\dots,M-1$ such that $\theta_{\rm DI_1}$ and $\theta_{\rm DI_2}$ are the departure angles from the transmitter and RIS, respectively, and  $\theta_{\rm AI}$ is  arrival at the RIS. Also, $\tilde{\nbh}\sim\ncalC\ncalN(0,{\rm I_N})$ and $\tilde{\nbH}_{:,i}\sim\ncalC\ncalN(0,{\rm I_N})$.

Let  $\mathbf{f}\in\nbbC^{N}$ and $\bm{\psi}\in\nbbC^{N}$ be the transmit beamforming vector and the RIS phase shift vector, respectively, and  let $\mathbf{\Phi}={\rm diag}(\bm{\psi})$ be the RIS phase shift  matrix. The elements of $\bm{\psi}$ are unit magnitude since RIS  is considered to be passive. 
Considering the far-field propagation scenario \cite{tang2020wireless}, the path loss for the direct and indirect links are modeled as $d_o^{-\alpha}$ and $(d_1d_2)^{-\alpha}$, respectively, where $\alpha$ is the path loss exponent and $d_o$, $d_1$, and $d_2$ are the  link distances of the transmitter-receiver, transmitter-RIS, and RIS-receiver links, respectively.

The signal received at a given time can be written as
\begin{align}
    y =  (d_1d_2)^{-\frac{\alpha}{2}}\nbh^T\mathbf{\Phi}\nbH\nbf x + d_o^{-\frac{\alpha}{2}}\nbg^T\nbf x + w,
\end{align}
where $x$ is the transmitted symbol with transmit power $\nbbE[x^Hx]=P_s$ and  $w$ is AWGN noise with variance $\sigma_w^2$.
Thus, the {\rm SNR} becomes
\begin{align}
    {\rm SNR}= \gamma|\nbh^T\mathbf{\Phi}\nbH\nbf+\mu \nbg^T\nbf|^2,\label{eq:SNR}
\end{align}
where $\gamma=(d_1d_2)^\alpha\frac{P_s}{\sigma_w^2}$ and $\mu=\left(\frac{d_1d_2}{d_o}\right)^{\frac{\alpha}{2}}$.

For the above setting, we aim to design the optimal transmit beamforming vector $\nbf$ and the RIS phase shift matrix $\bm{\Phi}$ so that the mean {\rm SNR} is maximized. For this, we formulate the optimization problem as 
\begin{subequations}
    \begin{align}
    \max_{\nbf,\bm{\psi}}~~ &\gamma\nbbE[|\nbh^T\mathbf{\Phi}\nbH\nbf+\mu \nbg^T\nbf|^2],\label{eq:objfun00}\\
    {\rm s.t.}~~~&\|\nbf\|=1,\label{eq:unitnorm0}\\
    &|\bm{\psi}_k|=1, ~~\forall k=0,\dots,N-1,\label{eq:unitmagnifitute0}%
\end{align}\label{SM_Opti_Prob_Def}%
\end{subequations}
where  \eqref{eq:unitnorm0}  represents the unit norm constraint for the transmit beamformer, and \eqref{eq:unitmagnifitute0} represents the unit magnitude constraint on the passive RIS elements. 

To characterize the performance of statistical CSI-based beamforming, we  will also analyze the {\em outage probability} and {\em ergodic capacity}. The outage probability is defined as the probability that the instantaneous {\rm SNR} is below threshold $\beta$ and is given by
\begin{align}
    {\rm P_{out}}(\beta)=\nbbP[{\rm SNR}\leq \beta].\label{eq:Pout}
\end{align}
The ergodic capacity is the enseble average of the {\rm SNR} and is given by ${\rm EC}=\nbbE[{\rm SNR}]$.
In the next section, we first obtain the optimal beamformer  that maximizes the mean {\rm SNR} and then we  analyze the outage  and ergodic capacity of the obtained beamforming solution. 
 
\section{Optimal Beamforming and Outage Analysis}
In this section, we present our two main contributions: 1) the joint optimization of the statistical-CSI-based beamformer and phase shift matrix and 2) the outage and capacity analyses. We formulate the optimization problem as given in \eqref{SM_Opti_Prob_Def} for maximizing the mean SNR under the beamformer and phase shift matrix constraints to obtain the optimal solution. However, this problem is non-convex due to the coupled beamformer and phase shift matrix variables. So, deriving a closed-form solution for this objective is challenging. Hence, we construct a lower bound on the mean SNR that is more conducive to analytical treatment and facilitates the derivation of the closed-form expressions for the optimal beamformer and phase shift matrix, as discussed next.

\subsection{Statistical CSI-based Beamforming}\label{subsec:optimalbeamforming}
In this subsection, we obtain the optimal transmit beamformer and RIS phase shift matrix that maximizes \eqref{SM_Opti_Prob_Def}.
The mean {\rm SNR} is obtained in   \cite[Appendix A]{hu2020statistical} as
\begin{align*}
    \nbbE[{\rm SNR}]&=\gamma |\kappa_l^2 \bar{\nbh}^T\mathbf{\Phi}\bar{\nbH}\nbf +\mu \kappa_l \bar{\nbg}^T\nbf|^2 + \gamma\kappa_l^2\kappa_n^2\|\bar{\nbH}\nbf\|^2\\
    &~~+\gamma  (\kappa_l^2\kappa_n^2+\kappa_n^4)N + \gamma \mu^2\kappa_n^2.
\end{align*}
Thus, we can rewrite the optimization problem \eqref{SM_Opti_Prob_Def} as
\begin{subequations}
\begin{align}
    \max_{\nbf,\bm{\psi}}~~ &|\kappa_l^2 \bar{\nbh}^T\mathbf{\Phi}\bar{\nbH}\nbf +\mu \kappa_l \bar{\nbg}^T\nbf|^2 + \kappa_l^2\kappa_n^2\|\bar{\nbH}\nbf\|^2,\label{eq:objfun}\\
    {\rm s.t.}~~~&\|\nbf\|=1,\label{eq:unitnorm}\\
    &|\bm{\psi}_k|=1, ~~\forall k=0,\dots,N-1\label{eq:unitmagnifitute}%
\end{align}    
\end{subequations}

For this optimization problem, the alternating transmit beamformer and RIS phase shifter subproblems-based solution is presented in \cite[Algorithm 1]{hu2020statistical}. 
Such numerical solutions do not provide insights into the exact functional dependence of the optimal solution on the key system parameters. 
Therefore,  closed-form expressions for beamformer $\nbf$ and RSI phase shift vector $\bm{\psi}$ are desirable. Towards this goal, we first obtain the optimal $\bm{\psi}$ for a given $\nbf$ in Subsection \ref{Subsection_Psi}, which we use to modify the objective function \eqref{eq:objfun} and solve for optimal $\nbf$ in Subsection \ref{Subsection_fI}.

\subsubsection{Optimal RIS Phase Shift Matrix:} \label{Subsection_Psi}
For a given $\nbf$, the optimization problem for the RIS phase shift matrix becomes 
\begin{subequations}
\begin{align}
    \max_{\bm{\psi}}~~ &|\kappa_l^2 \bm{\psi}^T{\rm diag}(\mathbf{\bar{\nbh}})\nbf +\mu \kappa_l \bar{\nbg}^T\nbf|^2  \label{OPT_RIS_OBJ},\\
    {\rm s.t.}~~~&|\bm{\psi}_k|=1, ~~\forall k=0,\cdots,N-1,\label{OPT_RIS_CON}%
\end{align}  \label{OPT_RIS_OPT} %
\end{subequations} 
where $\bar{\nbh}^T\mathbf{\Phi} = \bm{\psi}^T{\rm diag}(\mathbf{\bar{\nbh}})$. The above objective function can be upper-bounded as
\begin{align}
     |\kappa_l^2 \bm{\psi}^T{\rm diag}&(\mathbf{\bar{\nbh}})\nbf +\mu \kappa_l \bar{\nbg}^T\nbf|^2 \nonumber \\
     &\leq\left[\kappa_l^2 |\bm{\psi}^T{\rm diag}(\mathbf{\bar{\nbh}})\nbf| +\mu \kappa_l |\bar{\nbg}^T\nbf|\right]^2, \label{eq:inequality_1}
 \end{align}
where equality holds when $\bm{\psi}^T{\rm diag}(\mathbf{\bar{\nbh}})\nbf = c  \bar{\nbg}^T\nbf$ for a constant $c$. Thus, to achieve equality, we set 
 \begin{align}
 \bm{\psi}^T=c  \bar{\nbg}^T\nbf \nbw,    \label{eq:psi_equality}
 \end{align}
 where $\nbw=\nbf^H\nbE^H/\|\nbE\nbf\|^2$ is the pseudoinverse of $\nbE\nbf$ and $\nbE={\rm diag}(\bar{\nbh})\bar{\nbH}$. However, the constraint in \eqref{OPT_RIS_CON} also needs to be satisfied. Interestingly, we observe that the  elements of $\nbw$ have equal magnitudes and thus we can set $c$ to ensure \eqref{OPT_RIS_CON}. 
 
 Let $\nbe_n=\nbE_{n,:}$ be the $n$-th row of $\nbE$.
  By  construction of $\bar{\nbh}$ and $\bar{\nbH}$, we see that  $\nbe_n^H\nbe_n=\nbe_m^H\nbe_m$ and thus obtain
 \begin{align}
 \nbf^H\nbe_n^H\nbe_n\nbf=\nbf^H\nbe_m^H\nbe_m\nbf \Rightarrow \|\nbe_n\nbf\|^2=\|\nbe_m\nbf\|^2.\label{eq:enfI_abs}
 \end{align}
 In addition, we also observe that 
\begin{align}
  \|\nbE\nbf\|^2=\sum\nolimits_{n=0}^{N-1}|\nbe_n\nbf|^2=N|\nbe_n\nbf|^2.\label{eq:EfI_norm}
\end{align}
Using \eqref{eq:enfI_abs} and \eqref{eq:EfI_norm}, we get $$|\bar{\nbg}^T\nbf\nbw_k|=\frac{|\bar{\nbg}^T\nbf|}{N|\nbe_n\nbf|}.$$
where $|\nbw_k|=\frac{1}{N|\nbe_n\nbf|}$ for $\forall k$. 
Finally, by substituting the above equation in \eqref{eq:psi_equality}, we obtain the optimal RIS phase shift vector with unit magnitude elements as
  \begin{align}
      {\bm{\psi}^\star}^T=\frac{N|\nbe_n\nbf|}{|\nbg^T\nbf|}\nbg^T\nbf\nbw.\label{eq:psi_optimal}
  \end{align}

\subsubsection{Optimal Transmit Beamforming}\label{Subsection_fI}
The mean {\rm SNR} optimization problem to obtain the optimal transmit beamformer for the optimal $\bm{\psi}^{\star}$ is 
\begin{subequations}
\begin{align} 
    \max_{\nbf}~~ &|\kappa_l^2 \bm{\psi}^{\star^T}{\rm diag}(\mathbf{\bar{\nbh}})\nbf +\mu \kappa_l \bar{\nbg}^T\nbf|^2 + \kappa_l^2\kappa_n^2\|\bar{\nbH}\nbf\|^2,\label{OPT_fI_OBJ}\\
    {\rm s.t.}~~~&\|\nbf\|=1,\label{OPT_fI_Con}
\end{align}   \label{OPT_fI_OPT}  
\end{subequations}
We start by substituting $\bm{\psi}^\star$ in the first term of \eqref{OPT_fI_OBJ} as follows 
 \begin{align*}
     &|\kappa_l^2 {\bm{\psi}^\star}^T\nbE\nbf +\mu \kappa_l \bar{\nbg}^T\nbf|^2\\
     &=\kappa_l^4 |{\bm{\psi}^\star}^T\nbE\nbf|^2 +\mu^2 \kappa_l^2 |\bar{\nbg}^T\nbf|^2+2\kappa_l^3\mu |{\bm{\psi}^\star}^T\nbE\nbf||\bar{\nbg}^T\nbf|,\\
     &=N^2\kappa_l^4|\nbe_n\nbf|^2+\mu^2 \kappa_l^2 |\bar{\nbg}^T\nbf|^2+2N\kappa_l^3\mu |\nbe_n\nbf||\bar{\nbg}^T\nbf|,\numberthis\label{eq:equality_1}
 \end{align*}
where $\nbw\nbE\nbf=1$ follows from \eqref{eq:psi_equality}. 
The second term of \eqref{OPT_fI_OBJ} is simplified as
 \begin{align*}
     \|\bar{\nbH\nbf}\|^2&=\nbf^H\nbH^H\nbH\nbf,\\
     &=\nbf^H\nbH^H{\rm diag(\bar{\nbh})}^H{\rm diag(\bar{\nbh})}\nbH\nbf,\\
     &=\nbf^H\nbE^H\nbE\nbf,\\
     &=N|\nbe_n\nbf|^2,\numberthis\label{eq:HfI_norm}
 \end{align*}
where  the second and last equalities follow from ${\rm diag(\bar{\nbh})}^H{\rm diag(\bar{\nbh})}={\rm I_N}$ and \eqref{eq:EfI_norm}, respectively. Combining \eqref{eq:equality_1} and \eqref{eq:HfI_norm}, the objective \eqref{OPT_fI_OBJ} becomes
\begin{align*}
w_1|\nbe_n\nbf|^2+&w_2|\bar{\nbg}^T\nbf|^2+w_3|\nbe_n\nbf||\bar{\nbg}^T\nbf| = \\
&w_1\nbf^H\nbE_n\nbf+w_2\nbf^H\nbG\nbf+w_3|\nbf^H\nbE_g\nbf|,\numberthis\label{eq:modified_objfun}
\end{align*} 
 where $\nbE_n=\nbe_n^H\nbe_n$, $\nbG=\bar{\nbg}^\ast\bar{\nbg}^T$, $\nbE_g=\nbe_n^H\bar{\nbg}^T$,  $w_1=N^2\kappa_l^4+N\kappa_l^2\kappa_n^2$, $w_2=\mu^2\kappa_l^2$, and $w_3=2N\mu\kappa_l^3$. \newline
It is to be noted that $\nbE_n$ and $\nbG$ are  symmetric and positive semidefinite, whereas the matrix $\nbE_g$ is  negative definite. Thus, the presence of the third term in \eqref{eq:modified_objfun} makes the problem non-convex. For this reason, we ignore the last term from the maximization problem. This new objective will be equivalent to maximizing the lower bound on the mean {\rm SNR}.  It  is to be noted that this lower bound will be tight due to the following two reasons: 1) $w_1, w_2\gg w_3$, and 2) the eigenvalues of $\nbE_n$ and $\nbG$ are larger than the  $|\nbf^H\nbE_g\nbf|$. 

Using the above argument, we simplify the  optimization problem for  maximizing the lower bound of mean {\rm SNR} as 
\begin{subequations}
 \begin{align}
    \max_{\nbf}~~ & \nbf^H\nbZ\nbf,\label{eq:fIZfI}\\
    {\rm s.t.}~~~&\|\nbf\|=1,\label{OP21}
\end{align}\label{OP2}%
\end{subequations}
where  $\nbZ=w_1\nbE_n+w_2\nbG$ (symmetric matrix). This optimization is equivalent to the Rayleigh quotient maximization, whose solution is the dominant eigenvector of $\nbZ$. Finally, we summarize the optimal solution for the transmit beamforming vector and the phase shifter matrix  in the following theorem.
\begin{thm}
\label{Theorem_1}
The statistical CSI-based optimal transmit beamformer and RIS phase shift vector that maximizes the lower bound on the mean {\rm SNR}  (given in \eqref{eq:fIZfI}) are
\begin{align}
    \nbf^*=\nbv_1\text{~~~and~~~~} \bm{\psi}^\star=\frac{N|\nbe_n\nbf^\star|}{|\nbg^T\nbf^\star|}\nbw^T{\nbf^\star}^T\nbg,
\end{align}
respectively, where $\nbv_1$ is the principal eigenvector of $\nbZ$.
\end{thm} 
\subsection{Outage and Ergodic Capacity }
\label{subsec:outageanalysis}
In this subsection, we analyze the outage probability and ergodic capacity for the statistical CSI-based optimal beamformer $\nbf^\star$ and  RIS phase shift matrix $\bm{\psi}^\star$ obtained in Theorem \ref{Theorem_1}. 
For the given $\nbf^\star$ and $\bm{\psi}^\star$, the instantaneous {\rm SNR} becomes $${\rm SNR}=\gamma|\bm{\psi}^\star{\rm diag}(\nbh)\nbH\nbf^\star|.$$ 
For the above statistically lower bounded  {\rm SNR} expression, the outage probability and ergodic capacity respectively become
\begin{align}
{\rm P_{out}}(\beta)&=\nbbP[\gamma|\bm{\psi}^\star{\rm diag}(\nbh)\nbH\nbf^\star|\leq\beta],\\
\text{and~~}{\rm EC} &= \mathbb{E}[\log_2(1+\gamma|\bm{\psi}^\star{\rm diag}(\nbh)\nbH\nbf^\star|)].\label{eq:EC}    
\end{align}
A closed-form expression for this outage probability is derived in Theorem. \ref{Theorem_2}. Further, since we consider the proposed lower bound on mean {\rm SNR} for outage analysis, we will provide numerical verification on the efficacy of the derived outage probability in Section \ref{sec:numericalsection}.
\begin{thm}
\label{Theorem_2}
For the optimal transmit beamformer $\nbf^\star$ and phase shifter $\bm{\psi}^\star$ obtained in Theorem \ref{Theorem_1}, the channel gain follows a Rice distribution, which  gives outage probability as
\begin{align}
        {\rm P_{out}}(\beta)=1-Q_1\left(\frac{\nu}{\sqrt{2}\sigma},\frac{\sqrt{{\beta}/{\gamma}}}{\sqrt{2}\sigma}\right)\label{eq:outageprobability_thm2}
\end{align}
where $Q_1(\cdot)$ is a Marcum Q-function, and 
\begin{align*}
\nu&=N\kappa_l^2|\nbe_n\nbf^\star|+\mu\kappa_l|\bar{\nbg}^T\nbf^\star|,\\
\text{and}~~\sigma^2&= N\kappa_n^2(1+\kappa_l^2|\nbe_n\nbf^\star|^2)+\mu^2\kappa_n^2.
\end{align*}
\end{thm}
\begin{proof}
Please refer to Appendix \ref{app:outageprobability} for the proof.
\end{proof}
Using Theorem \ref{Theorem_2}, we can evaluate the ergodic capacity given in \eqref{eq:EC}  as done in the following corollary.
\begin{cor}
The ergodic capacity of the optimal transmit beamformer and RIS phase shifter given in Theorem \ref{Theorem_1} is 
\begin{align}
    {\rm EC}=\frac{1}{\ln(2)}\int_0^\infty \frac{1}{1+u}Q_1\left(\frac{\nu}{\sqrt{2}\sigma},\frac{\sqrt{{u}/{\gamma}}}{\sqrt{2}\sigma}\right){\rm d}u.
\end{align}
\end{cor}

In the following section, we present the numerical verification for the efficacy of the proposed beamforming scheme and its outage probability, and we also discuss its capacity performance under various system settings.   
\section{Numerical Results and Discussion}
\label{sec:numericalsection}
In this section, we present a numerical performance analysis of the proposed beamforming scheme for maximizing the lower bound of mean SNR. We refer to the proposed scheme as {\em Max LB Mean SNR}. We also compare the performance of our proposed scheme with the scheme maximizing exact mean {\rm SNR} using statistical {\rm CSI} \cite{hu2020statistical} and the  scheme maximizing instantaneous SNR using perfect {\rm CSI}. We refer to the former scheme as {\em Max Mean SNR} and the latter one as {\em Max SNR}.
The parameters $\gamma$ and $\mu$, defined below \eqref{eq:SNR}, represent the received indirect link {\rm SNR}  and the square root of the ratio of received SNRs over direct and indirect links, respectively, observed under the SISO channel. Besides,  $\mu$ is crucial for optimizing the transmit beamformer, as can be verified using \eqref{eq:modified_objfun} and the discussion below it. 
To highlight this, we  present the numerical results for various values of $\gamma$,  $\mu$, and the key system parameters.   
For the numerical analysis, we consider $M = 4, N = 32, \theta_{\rm D} = 0, \theta_{\rm DI1} = \frac{1}{4}\pi, \theta_{\rm DI2} = \frac{8}{5}\pi, K = 5, \alpha = 3.5, \gamma=0~\rm{dB},~\text{and}~\mu = 5~\rm{dB}$; unless mentioned otherwise.  
\begin{figure}[h]
\centering
 \includegraphics[width=0.5\textwidth]{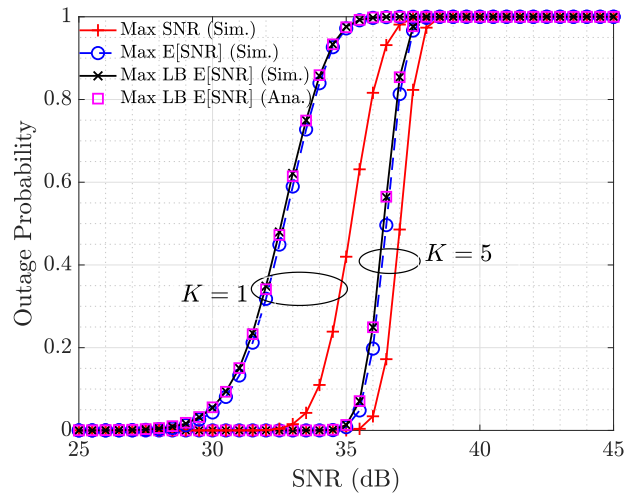} \vspace{-.5cm}
 \caption{Outage Probability.} 
 \label{fig:OP}
\end{figure}

Fig. \ref{fig:OP} shows the outage probability for the proposed and above-mentioned beamforming schemes. As expected, we notice that the outage probability derived for the proposed scheme in Theorem \ref{Theorem_2} matches with the simulation. Further, this result  demonstrates that the outage  of the proposed scheme is a tight upper bound to the outage of the {\em Max Mean SNR} scheme. This is expected as the proposed scheme maximizes the lower bound of the mean {\rm SNR}.  The fact that the proposed scheme is a tight bound also verifies the efficacy of the proposed  {\em Max LB Mean SNR} scheme in solving the original optimization problem. Besides, it can be seen that the {\em Max SNR} scheme outperforms the other two schemes for obvious reasons. However, the performance gap reduces with the increase of Rice fading factor $K$, which is also deducible from Fig. \ref{fig:CapVsN}.  
 
\begin{figure}[h]
\centering
 \includegraphics[width=0.5\textwidth]{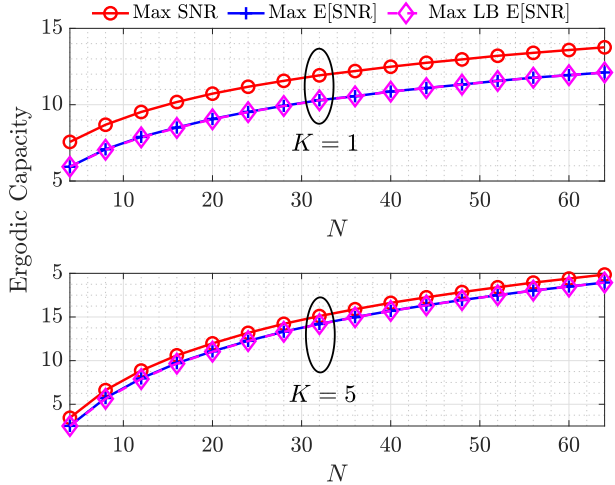} 
 \caption{ Ergodic capacity vs.  number of RIS elements $N$. } 
 \label{fig:CapVsN}
\end{figure}
Fig. \ref{fig:CapVsN} shows that the capacity increases with the increase in the number of RIS reflecting elements $N$, which is expected. The figure also verifies that the proposed {\em Max LB Mean SNR} scheme provides a close lower bound to the capacity achieved via {\em Max Mean SNR} for all values of $N$.  As noticed in the outage result, the ergodic capacity for {\em Max SNR} scheme  also outperforms the statistical-CSI-based schemes for smaller values of $K$. 
Fig. \ref{fig:CapVslambda} shows the capacity increases with the increase of parameter $\mu$. Given $\gamma$, increasing $\mu$ is equivalent to increasing the {\rm SNR} of the direct link which increases the overall capacity.  Besides, the figure shows that the  capacity also increases  with the increase in the parameter $\gamma$, however the observed increment is independent of $\mu$ and also other parameters (which can also be seen from \eqref{eq:EC}).

Fig. \ref{fig:CapVstdBR} presents the ergodic capacity as a function of the difference of departure angles of direct and indirect links from the transmitter, i.e., $\theta = \theta_{\rm DI_1} - \theta_{\rm DD}$. It can be seen that the capacity decreases from its maximum value at $\theta = 0$ to its minimum value at $\theta = \frac{\pi}{2}$. This is expected as both the direct and indirect links get separated with the increase of $\theta$ for which  it becomes difficult to form a narrow transmit beam and hence the gain reduces. However, it can be observed that the capacity gap between {\em Max Mean SNR} scheme and the proposed {\em Max LB Mean SNR} scheme increases sightly for higher $\theta$ and $\mu$. This may be attributed to the fact that the third term (ignored for lower bounding the SNR) in \eqref{eq:HfI_norm} becomes somewhat unavoidable. However, the gap is very small.     
\begin{figure}[h]
\centering
 \includegraphics[width=0.5\textwidth]{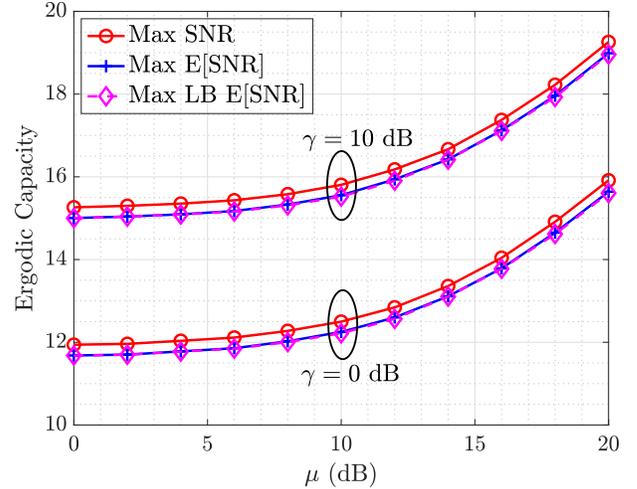} 
 \caption{Ergodic capacity with respect to $\mu$.} 
 \label{fig:CapVslambda}
\end{figure}
\begin{figure}[h]
\centering
 \includegraphics[width=0.5\textwidth]{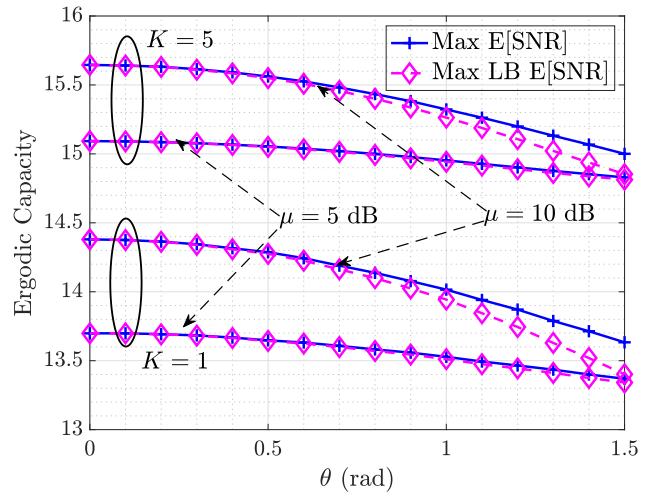} 
 \caption{ Ergodic Capacity vs. the difference of departure angles indirect and direct links.} 
 \label{fig:CapVstdBR}
\end{figure}
\section{Conclusion}
In this paper, we designed a statistical CSI-based optimal beamforming scheme and characterized its outage and capacity performances for a RIS-aided MISO communication system when both direct and indirect link experience Rician fading. Unlike prior art where such problems are often tackled numerically, we derived closed-form expressions for the optimal transmit beamformer and the optimal RIS phase shift matrix that maximize a tight lower bound on the SNR. Besides, we also show that the effective channel gain follows Rice distribution under the proposed beamforming scheme, which immediately leads to a closed-form result for the outage probability. Subsequently, we use this result to derive an analytical expression for the ergodic capacity of the proposed beamforming scheme. Using numerical comparisons, we demonstrate that the proposed closed-form solution acts as a tight upper bound for the numerically-obtained solution for the exact mean SNR maximization-based beamforming scheme. Likewise, the achievable capacity of the proposed beamformer is shown to be a close lower bound to the capacity of the mean SNR maximization-based beamforming scheme.
\appendix 
\label{app:outageprobability} 
Let us define $\xi_1={\bm{\psi}^\star}^T  \rm{diag}(\nbh) \nbH \nbf^\star$ and $\xi_2=\nbg^T\nbf^\star$ and rewrite {\rm SNR}  given in \eqref{eq:SNR} as 
\begin{align*}
   \gamma |{\bm{\psi}^\star}^T  \rm{diag}(\nbh) \nbH\nbf^\star+\mu \nbg^T\nbf^\star|^2 = \gamma |\xi_1+\mu\xi_2|^2,
\end{align*}
Thus, outage can be evaluated as
\begin{align}
    {\rm P_{out}}(\beta)=\nbbP[|\xi_1+\mu\xi_2|\leq \sqrt{\beta/\gamma}]\label{eq:outageprobability_2}.
\end{align}
Using channel models from \eqref{eq:g_model} and \eqref{eq:hH_model}, we obtain
\begin{align}
    \xi_1&=\kappa_l^2a_1 + \kappa_l\kappa_n a_2 +\kappa_l\kappa_n a_3+ \kappa_n^2 a_4,\label{eq:xi1}\\
    \xi_2 &= \kappa_l b_1 + \kappa_n b_2, \label{eq:xi2}
\end{align}
\begin{align*}
\text{where}~~a_1&={\bm{\psi}^\star}^T{\rm diag}(\bar{\nbh})\bar{\nbH}\nbf^\star,~~ &a_2={\bm{\psi}^\star}^T{\rm diag}(\bar{\nbh})\tilde{\nbH}\nbf^\star,\\ a_3&={\bm{\psi}^\star}^T{\rm diag}(\tilde{\nbh})\bar{\nbH}\nbf^\star, ~~ &a_4={\bm{\psi}^\star}^T{\rm diag}(\tilde{\nbh})\tilde{\nbH}\nbf^\star,\\   
b_1&=\bar{\nbg}^T\nbf^\star,~~
&b_2=\tilde{{\nbg}}^T\nbf^\star.~~~~~~~~~~~~~\vspace{.1mm}
\end{align*}
We now find distributions of   $\tilde{\nbp}=\bm{\psi}^\star{\rm diag}(\tilde{\nbh})$, $\tilde{\nbq}=\tilde{\nbH}\nbf^\star$, and $\tilde{u}=\tilde{\nbg}^T\nbf^\star$. These results will be used to identify the distributions of $a_i$s and $b_i$s, which subsequently will be employed to obtain the distributions of the channel gain and {\rm SNR}.  Using the rotational invariance property of zero-mean complex Gaussian distribution, we obtain $\tilde{\nbp}\sim\ncalC\ncalN(0,{\rm I_N})$. Further, by using the linear combination property of zero-mean complex Gaussian random variable and the fact that $\|\nbf^\star\|=1$, we find that $\tilde{\nbq}\sim\ncalC\ncalN(0,{\rm I_N})$ and $\tilde{u}\sim\ncalC\ncalN(0,1)$.

\noindent{\em Terms $a_1$ and $b_1$:} The  terms $a_1$ and $b_1$ are constants as given above. The term $a_1$ can be simplified as 
\begin{align}
    a_1=\frac{N|\nbe_n\nbf^\star|}{|\bar{\nbg}^T\nbf^\star|}\bar{\nbg}^T\nbf^\star.
\end{align}

\noindent{\em Term $a_2$:} Since $\nbp={\bm{\psi}^\star}^T{\rm diag}(\bar{\nbh})$ is a complex vector of unit-magnitude elements, we can show that $a_2=\nbp\tilde{\nbq}$ is a zero-mean complex Gaussian random variable with variance   
\begin{align*}
    {\rm Var}[a_2] = \nbbE[\nbp\tilde{\nbq}\tilde{\nbq}^H\nbp^H]=\nbp\nbbE[\tilde{\nbq}\tilde{\nbq}^H]\nbp^H=\nbp\nbp^H=N.
\end{align*}\\
{\em Term $b_2$:} Note  $b_2=\tilde{u}$ is a zero-mean unit variance complex Gaussian, as mentioned above. \\[2mm]
{\em Term $a_3$:}
With $\nbq=\bar{\nbH}\nbf^\star$, we can show that $a_3=\tilde{\nbp}\nbq$ is also a zero-mean complex Gaussian with variance 
\begin{align*}
    {\rm Var}[a_3]=\nbbE[\nbq^H\tilde{\nbp}^H\tilde{\nbp}\nbq]=\nbq^H\nbbE[\tilde{\nbp}^H\tilde{\nbp}]\nbq=\nbq^H\nbq=\|\bar{\nbH}\nbf^\star\|^2.
\end{align*}

\noindent{\em Term $a_4$:} Using similar arguments, we can show that $a_4=\tilde{\nbp}\tilde{\nbq}$ follows the complex Gaussian with zero mean and variance
\begin{align*}    \nbbE[a_4^Ha_4]&=\nbbE[\tilde{\nbq}^H\tilde{\nbp}^H\tilde{\nbp}\tilde{\nbq}],\\   &=\nbbE[\tilde{\nbq}^H\nbbE[\tilde{\nbp}^H\tilde{\nbp}]\tilde{\nbq}],\\ 
    &=\nbbE[\tilde{\nbq}^H\tilde{\nbq}],\\
    &=\sum\nolimits_{n=0}^{N-1} \nbE[|\tilde{\nbq}_n|^2],\\
    &=N.
\end{align*}

\noindent{\em Distribution of $\xi_1$ and $\xi_2$:} Compiling the above results gives the distributions of $\xi_1$  and $\xi_2$ (given in \eqref{eq:xi1} and \eqref{eq:xi2}) as
\begin{align*}
    \xi_1&\sim\ncalC\ncalN\left(\kappa_l^2\frac{N|\nbe_n\nbf^\star|}{|\bar{\nbg}^T\nbf^\star|}\bar{\nbg}^T\nbf^\star,N\kappa_n^2(1+\kappa_l^2|\nbe_n\nbf^\star|^2)\right),\\
    \xi_2&\sim\ncalC\ncalN\left(\kappa_l\bar{\nbg}^T\nbf^\star,\kappa_n^2\right).
\end{align*}
Using this, we get 
\begin{align*}
    \xi_1+\mu\xi_2&\sim\ncalC\ncalN(m,\sigma^2)
\end{align*}
where $m=\kappa_l^2\frac{N|\nbe_n\nbf^\star|}{|\bar{\nbg}^T\nbf^\star|}\bar{\nbg}^T\nbf^\star+\mu\kappa_l\bar{\nbg}^T\nbf^\star$ and 
$\sigma^2=N\kappa_n^2(1+\kappa_l^2|\nbe_n\nbf^\star|^2)+\mu^2\kappa_n^2$.\\[2mm]
{\em Channel gain distribution:} We note that the magnitude of  a nonzero-mean complex Gaussian follows the Rice distribution. Hence, we obtain $$|\xi_1+\mu\xi_2|\sim{\rm Rice}(\nu,\sigma/\sqrt{2}),$$ where
$\nu=|m|$. Finally, using \eqref{eq:outageprobability_2} and the {\rm CDF} of Rice distribution, the outage probability is obtained as  in \eqref{eq:outageprobability_thm2}.

\bibliographystyle{IEEEtran} 
\bibliography{Reference.bib}
\end{document}